\documentclass[12pt,letterpaper]{article}

\usepackage{fixltx2e}
\usepackage{textcomp}
\usepackage{fullpage}
\usepackage{amsfonts}
\usepackage{verbatim}
\usepackage[english]{babel}
\usepackage{pifont}
\usepackage{color}
\usepackage{setspace}
\usepackage{lscape}
\usepackage{indentfirst}
\usepackage[normalem]{ulem}
\usepackage{booktabs}
\usepackage{natbib}
\usepackage{float}
\usepackage{latexsym}
\usepackage{url}
\usepackage{hyperref}
\usepackage{epsfig}
\usepackage{graphicx}
\usepackage{amssymb}
\usepackage{amsmath}
\usepackage{bm}
\usepackage{array}
\usepackage{mhchem}
\usepackage{ifthen}
\usepackage{caption}
\usepackage{hyperref}
\usepackage{amsthm}
\usepackage{amstext}
\usepackage{array}
\newcolumntype{L}[1]{>{\raggedright\let\newline\\\arraybackslash\hspace{0pt}}m{#1}}

\raggedright
\renewcommand{\section}[1]{%
\bigskip
\begin{center}
\begin{Large}
\normalfont\scshape #1
\medskip
\end{Large}
\end{center}}

\renewcommand{\subsection}[1]{%
\bigskip
\begin{center}
\begin{large}
\normalfont\itshape #1
\end{large}
\end{center}}

\renewcommand{\subsubsection}[1]{%
\vspace{2ex}
\noindent
\textit{#1.}---}

\renewcommand{\tableofcontents}{}

\bibpunct{(}{)}{;}{a}{}{,}  % this is a citation format command for natbib

\begin{document}
\begin{flushright}
Version dated: \today
\end{flushright}
\bigskip
\noindent WiSPA: A new approach for dealing with widespread parasitism
% put in your own RH (running head)
% for POVs the RH is always POINT OF VIEW

\bigskip
\medskip
\begin{center}

% Insert your title:
\noindent{\Large \bf WiSPA: A new approach for dealing with widespread parasitism}
\bigskip

% We don't use a special title page; the author information is entered 
% like any other text.

% FOOTNOTES: We don't allow them in the manuscript, except in
% tables. Don't include any footnotes in the text.

\noindent {\normalsize \sc Benjamin Drinkwater$^1$, Angela Qiao$^1$, and Michael A. Charleston$^1$ $^2$}\\
\noindent {\small \it 
$^1$School of Information Technologies, University of Sydney, NSW, 2006, Australia;\\$^2$School of Physical Sciences, University Of Tasmania, TAS, 7005, Australia;}
\end{center}
\medskip
\noindent{\bf Corresponding author:} Benjamin Drinkwater, School of Information Technologies (J12), University of Sydney, NSW, 2006, Australia; E-mail: benjamin.drinkwater@sydney.edu.au\\

% Of course the specific format of addresses may vary according to
% country or other factors. Also, that was just an example email format.
%It's acceptable to add email addresses for authors in addition to the
%corresponding author. These would be placed after "Country."

\vspace{1in}

\subsubsection{Abstract}
Traditionally, studies of coevolving systems have considered cases where a parasite may inhabit only a single host. The case where a parasite may infect many hosts, \emph{widespread parasitism}, has until recently gained little traction. This is due in part to the computational complexity involved in reconstructing the coevolutionary histories where parasites may infect only a single host, which is NP-Hard. Allowing parasites to inhabit more than one host has been seen to only further compound this computationally intractable problem. Recently however, well-established algorithms for estimating the problem instance where a parasite may infect only a single host have been extended to handle widespread parasites. Although this has offered significant progress, it has been noted that these algorithms poorly handle parasites that inhabit phylogenetically distant hosts. 

In this work we extend these previous algorithms to handle cases where parasites inhabit phylogenetically distant hosts using an additional evolutionary event which we call \emph{spread}. Our new framework is shown to infer significantly more congruent coevolutionary histories compared to existing methods over both synthetic and biological data sets. We then apply the newly proposed algorithm, which we call WiSPA (WideSpread Parasitism Analyser), to the well studied coevolutionary system of \emph{Primates} and \emph{Enterobius} (pinworms), where existing methods have been unable to reconcile the widespread parasitism present without permitting additional divergence events. Using WiSPA and the new biological event, spread, we provide the first statistically significant coevolutionary hypothesis for this system. \\
\noindent (Keywords: Coevolution, Phylogeny, Widespread Parasitism, NP-Hard )\\

% Points of View do not have abstracts but they should include
% Keywords.

\vspace{1.5in}

Coevolutionary research has long focused on the area of parasitism due to the health risks which parasites pose to the human population \citep{charleston2006traversing}. Parasites and the associations they form with their hosts have been responsible for a number of the worst emerging diseases impacting global health today, including \emph{Ebola} \citep{peterson2004ecologic}, HIV \citep{siddall1997aids}, and malaria \citep{mu2005host}. Further research into the field of coevolution aims to uncover the deep coevolutionary associations formed by parasitic behaviour, to provide further insights into these deadly diseases \citep{charleston2006cophylogenetic}. 

We often define coevolutionary systems in terms of an independent phylogeny and a corresponding dependent phylogeny which have formed a macro-scale coevolutionary bond. One approach that is often applied to evaluating such evolutionary relationships is the field of \emph{cophylogenetics}, which provides a framework to evaluate whether evolutionary histories have coevolved or have evolved independently \citep{charleston2003recent}. 

As a result of host--parasite systems' long association with the field of cophylogenetics, coevolutionary systems often describe the independent and dependent phylogenies as the \emph{host} $(H)$ and \emph{parasite} $(P)$ respectively. Cophylogenetic analysis, however, can be applied to all forms of coevolutionary dependence including: biogeography \citep{toit2013biogeography}, host--pathogen systems \citep{mu2005host}, genes and the species that house them \citep{page1997gene}, plant--insect interactions \citep{gomez2010neotropical}, plant--fungi dynamics \citep{refregier2008cophylogeny}, host--parasitoid relationships \citep{stireman2005evolution}, and Batesian and M\"{u}llerian mimicry between species \citep{ceccarelli2007dynamics,cuthill2012phylogenetic}.

The coevolutionary interactions between $P$ and $H$ are represented by the associations $(\varphi)$ between their leaves, based on evidence of parasites inhabiting or infecting their host(s). These associations can be used to infer the level of host specificity of the parasite species with respect to its host(s) \citep{poulin2011evolutionary}. Within this context, high host specificity is the case where a particular parasite infects a single host species, while low host specificity is the case where a parasite may infect many host species.

Coevolutionary analysis of systems with high host specificity focuses on the reconstruction of the parasite's evolutionary history with respect to the host, which is known as a \emph{cophylogeny mapping}. When recovering a map $(\Phi)$ using cophylogeny mapping, the aim is to recover the most congruent solution with a minimum total event cost, while ensuring the associations are conserved using the four known coevolutionary events, \emph{codivergence}, \emph{duplication}, \emph{host switch} and \emph{loss} \citep{ronquist1995reconstructing}. 

A codivergence event is a concurrent divergence of both the host and parasite lineages. A high concentration of codivergence events leads to an increase in the level of congruence between $P$ and $H$, and is therefore a strong indicator of coevolution \citep{page2002tangled}. A duplication event is an independent divergence of the parasite where both new lineages continue to track the host \citep{tuller2010reconstructing}. A host switch event is an independent divergence of the parasite where one parasite shifts from the initial host lineage (take--off edge) to a new lineage (landing edge) in the host, while the second parasite continues to track the host \citep{kim1985coevolution}. We call these three events \emph{divergence events}, as they consider all cases of divergence in the parasite's coevolutionary history. By contrast, loss arises from three indistinguishable processes: lineage sorting (or ``missing the boat''), extinction, or sampling failure. As these processes all produce the same effect we represent these as a \emph{loss} event \citep{paterson2003drowning}. We refer to the problem of reconstructing a map using only these four events as the \emph{restricted cophylogeny reconstruction problem}.

Methods for recovering maps have mainly focused on the restricted cophylogeny reconstruction problem. This is due in part to the initial set of biological events being unable to reconstruct the evolutionary history of parasites with low host specificity, along with the hypothesis that coevolution only occurs in systems with a one-to-one association between parasites and their hosts \citep{poulin2011evolutionary}. This hypothesis, however, only considers a select set of coevolving systems and precludes many observed coevolutionary systems where the parasites maintain low host specificity as an evolutionary advantage. In a comprehensive study of plant--insect interactions, \citeauthor{nosil2005testing} \citeyearpar{nosil2005testing} demonstrated that while insects often form exclusive associations with their hosts, this is not always the case. Butterflies and bark beetles were shown to often be associated with many host plant species. This case is not unusual, with \citeauthor{stireman2005evolution}'s \citeyearpar{stireman2005evolution} study of endoparasitoids and their \emph{Tachinidae} (fly) hosts also demonstrating the evolutionary advantage of low host specificity. These results affirm that ongoing cophylogeny mapping modelling must consider the general case where parasites are permitted to inhabit more than one host (\emph{widespread parasitism}), to accurately model all coevolutionary interrelationships.

As described above, modelling widespread parasitism using cophylogeny mapping requires additional biological events beyond the original four events derived by \citeauthor{ronquist1995reconstructing} \citeyearpar{ronquist1995reconstructing}. Currently, failure-to-diverge is the only event which has successfully been applied to handle widespread parasitic events within a cophylogeny mapping framework. It is defined as the case where parasites maintain their ability to inhabit both hosts following a host divergence event, without the need for a divergence of the parasite lineage \citep{johnson2003parasites}. Failure-to-diverge is the case where there is an interruption of the gene-flow between the host species while there remains gene flow within the parasite population \citep{poulin2011evolutionary}. A case where failure-to-diverge and the full set of divergence events are required to reconstruct a cophylogenetic history can be seen in Figure~\ref{fig:TanglegramAndMap}. We will here refer to events which are used to describe widespread parasite coevolution, such as failure-to-diverge, as \emph{widespread events}. This is to differentiate such events from divergence and loss events.

The failure-to-diverge event allows for the recovery of solutions for all conceivable cases of widespread parasitism for cophylogenetic reconstructions. However, these solutions may have a high number of loss events when widespread parasites inhabit phylogenetically distant leaves in $H$. This is due to the limitation that a failure-to-diverge event occurs at the most recent common ancestor of the pair of inhabited host leaves \citep{banks2005multi}, after which many loss events must be inferred to account for the observed parasite distribution.

Cophylogenetic reconstructions are evaluated using an event cost, similar to that of a parsimony score in phylogenetic reconstructions \citep{charleston2002principles}. Reconstructing the minimum cost map requires that each divergence event, widespread event, and loss event be assigned a penalty cost. The set of costs for each event may be defined as a vector $V = (C,D,W,L,F)$ where $C$, $D$, $W$, $L$, $F$ represent the associative costs for each codivergence, duplication, host switch, loss, and failure-to-diverge respectively. The resultant map cost $E$ can then be derived as:

\begin{equation}\label{eq:costFunction}
E = \alpha C + \beta D + \gamma W + \delta L + \epsilon F
\end{equation}

\noindent where $\alpha$, $\beta$, $\gamma$, $\delta$, $\epsilon$ represent the number of events for codivergence, duplication, host switch, loss, and failure-to-diverge respectively \citep{drinkwater2014improved}.

Cophylogeny mapping algorithms aim to map $P$ into $H$, where the number of codivergence events is maximised and the map cost $E$ is minimised \citep{charleston2014event}. Although such cost schemes may not evaluate coevolutionary scenarios exactly, particularly modelling preferential host switching \citep{charleston2002preferential}, this technique has been used to evaluate a large number of coevolutionary systems \citep{page2002tangled,page2004phylogeny,jackson2004cophylogenetic,cruaud2012extreme,rivera2015lineage}.

Parsimony and event-based methodologies are often seen as less preferable to maximum likelihood methodologies. This is due to parsimony methods often relying on arbitrarily chosen cost schemes. While such a reliance is a limiting factor, parsimony and event based methods may be used to reconstruct the most likely evolutionary history by assigning the negative log likelihood probabilities for each evolutionary event as the associated penalty cost for each event \citep{drinkwater2016rascal}. As a result there is a strong driver for fast mapping methods which can then be integrated into a coevolutionary likelihood framework \citep{charleston2003recent}, to complement maximum likelihood techniques \citep{baudet2015cophylogeny}.

Unfortunately recovering the minimum cost map is known to be NP-Hard \citep{ovadia2011cophylogeny}. This computational intractability is due to the exponential number of host switch locations that can arise due to the variable order of the internal nodes in the host tree \citep{doyon2011efficient}, and the exponential number of internal node orderings \citep{conow2010jane}.

To mitigate this computational intractability two unique heuristics have been proposed. The first approach ignores the relative ordering of the internal nodes in the parasite phylogeny, This may lead to the order of evolutionary events, as defined by a reconciled map, contradicting the order of evolutionary events as defined by the parasite phylogeny \citep{doyon2011efficient}. Such a map is often referred to as \emph{biologically infeasible} or time-inconsistent \citep{doyon2011models}. This approach has been applied in various methods \citep{merkle2005reconstruction,merkle2010parameter,yodpinyaneehmc}, with the fastest known algorithm to date running in $O(n^2)$ \citep{bansal2012efficient}.

To guarantee that solutions are time-consistent two properties must be ensured. A host switch's take--off and landing edges must lie in the same time interval, which is an overlapping interval based on the edges' distances from the root of $H$, and must maintain the partial ordering of $P$ \citep{conow2010jane}, as mentioned above. This requires that the relative order of the parasite phylogeny must be fixed, which has led to the second heuristic which fixes the internal node ordering of the host phylogeny. If the internal node ordering is fixed it is possible to solve the cophylogeny reconstruction problem in polynomial time. This simplified problem, often referred to as the dated tree reconciliation problem \citep{drinkwater2015subquadratic}, has been applied within a number of algorithms \citep{doyon2011models,conow2010jane,drinkwater2015time}, with the fastest proposed to date running in $O(n^2\log{n})$ \citep{bansal2012efficient}.

While the aim of the cophylogeny reconstruction problem is to recover the minimum cost map in terms of, $E$, it is often valuable to infer the significance of the resultant map. In particular it is valuable to identify if a resultant map provides a statistically significant signal that the apparent congruence is unlikely to have occurred simply by chance. Previously, analysis has applied a series of Bernoulli trials to analyse the significance of an inferred map, such as the analysis of pocket--gophers and their parasitic chewing lice by \citeauthor{page1994maps} in \citeyear{page1994maps}. This is also the premise of the statistical evaluation tool Parafit \citep{legendre2002statistical} which randomises the associations or the parasite tree to identify whether the degree of congruence noted between the host and parasite tree could have occurred simply by chance. This process can be replicated when using cophylogeny mapping, by producing randomised permutations of the initial tanglegram (either randomising the associations or the parasite tree), and computing the cost of the optimal map for this randomised instance \citep{page1994maps,page1994parallel}. If the initial tanglegram has a mapping cost, $E$, which is less than the randomised permutations in at least 95\% of cases, then we may reject the null hypothesis that there is no significance between the independent and dependent phylogenetic trees. This feature is currently integrated into the most recent implementation of the Jane software tool \citep{conow2010jane}.

Three recent software tools which have been designed to recover maps where widespread parasites are considered and using the failure-to-diverge widespread evolutionary event, are CoRe-PA \citep{merkle2010parameter}, Jane \citep{conow2010jane} and CoRe-ILP \citep{wiesekecophylogenetic}. CoRe-PA solves the cophylogeny reconstruction problem in polynomial time by relaxing the internal node ordering of the host tree. This approach, although potentially recovering solutions in quadratic time \citep{yodpinyaneehmc}, may recover solutions which are time-inconsistent \citep{doyon2011models}. Jane, in contrast, fixes the internal node order in the host tree and solves this instance using dynamic programming \citep{libeskind2009computational}. This approach guarantees that solutions are biologically feasible. As there are an exponential number of possible fixed node orderings, most techniques applying this approach leverage a genetic algorithm to recover the best possible solutions in a fixed period of time. Finally, CoRe-ILP applies an integer linear programming algorithm to solve the problem of maximising the total number of codivergence events within the reconciled map, which its developers have shown provides a robust estimation of the harder problem of finding the minimum cost map \citep{wiesekecophylogenetic}. 

While each approach handles the computational intractability of recovering the divergence events in significantly different ways, all apply a common approach for recovering widespread events, where each failure-to-diverge event occurs at the most recent common ancestor of the parasite's host, guaranteeing solutions can be recovered in all cases. This approach while offering researchers the first set of tools for inferring coevolutionary systems which include widespread parasites, often infers maps with a high number of loss events polluting the coevolutionary signal. 

Our work aims to expand on this research by constructing a new methodology which solves the cophylogeny reconstruction problem with widespread parasites, the \emph{widespread parasite problem}, which is able to overcome the high costs that are often associated with failure-to-diverge. As a result, this method will decrease the overall parsimony score, while potentially increasing the number of codivergence events within the reconciled map.

\section{Methodology}

\subsection{Reintroducing an additional biological event for coevolutionary analysis}

This work reintroduces an additional evolutionary event for inferring relationships of coevolutionary systems where widespread parasites are permitted, which we call \emph{spread}. We propose that existing frameworks be updated to include spread as an additional widespread parasite event, to work in conjunction with failure-to-diverge. The inclusion of spread aims to more accurately reconcile the widespread parasites' coevolutionary histories with respect to their hosts, by mitigating the high number of loss events which are associated with reconstructions that exclusively use failure-to-diverge, such as cases where parasites \emph{do not} inhabit closely related hosts.

The spread event was first applied by \citeauthor{brooks1991phylogeny}, \citeyearpar{brooks1991phylogeny}, to reconcile widespread parasites within the \citeauthor{brooks1991phylogeny} Parsimony Analysis framework, and was later proposed by \citeauthor{siddall2003brooks}, \citeyearpar{siddall2003brooks}, as an additional widespread evolutionary event to be integrated within TreeMap \citep{TreeMapWebsite}. In both cases, however, neither of these proposed models have been implemented in part due to the additional computational complexity that their inclusion can give rise to.

The spread event is derived from a number of observed parasitic systems, such as the behaviour of chewing lice which infect their penguins hosts \citep{banks2005cophylogenetic}. It has been observed that a number of lice species switch between their penguin hosts at shared breeding grounds, however, this cannot be modelled using a host switch, as there is no divergence in the parasite lineage, nor can this be considered a failure-to-diverge event as there is no evidience to support that the gene flow has been maintained for the chewing lice species. Rather, the lice species have recently spread to new hosts based on new opportunities that are presented.

This ``spreading'' behaviour has also been observed in lab experiments between nematodes and their \emph{Drosophila} fly hosts \citep{jaenike1998general}. Nematodes' general purpose genotypes allow each individual species to infect a high number of host species, allowing nematodes to infect distantly related \emph{Drosophila} hosts which they would not be expected to encounter in nature. The infection of \emph{Drosophila} does not require any evolutionary changes so this cannot be modelled correctly using a host switch event, nor can it be described using failure-to-diverge as the nematodes have not coexisted with their new hosts, and therefore this phenomenon requires an additional biological event to model this observed behaviour, which spread successfully achieves.

The spread event also complements the theory of low host specificity which asserts species evolve specific mechanisms which allow them to inhabit multiple hosts where the parasites are less vulnerable to the evolutionary changes of a specific host species. Further, spread often provides more parsimonious solutions for widespread parasitism. Consider the coevolutionary system in Figure~\ref{fig:introducingSpread} (left). In the first reconstruction, Figure~\ref{fig:introducingSpread} (center), the parasite has had $O(n)$ opportunities to infect a new host species but failed to do so in all cases. This is highly unlikely compared to the alternate reconstruction using spread, where no loss events occur and the parasite simply infects a new host species based on new opportunities, such as the introduction of an infected host species $(A)$ into host species $(B)$ natural environment. 

The alternate map which uses spread, Figure~\ref{fig:introducingSpread} (right), is significantly more parsimonious for cases where the spread event is assigned a penalty cost similar to that of failure-to-diverge. In fact, spread would need to be assigned a cost $n$ times that of a loss event for the solution to be considered more expensive. Therefore, as this biological event describes observed behaviour in nature and also allows for potentially more parsimonious maps, we argue that spread should be integrated into existing algorithms which aim to infer systems which present widespread parasites. This is in line with previous assertions made by \citeauthor{brooks1991phylogeny}, \citeyearpar{brooks1991phylogeny}, \citeauthor{page1994maps}, \citeyearpar{page1994maps}, and \citeauthor{siddall2003brooks}, \citeyearpar{siddall2003brooks}.

While often producing significantly cheaper solutions, spread may not always be possible as it is reliant on host species collocation to permit the occurrence of a spread event similar to host switch events \citep{clayton2004ecology}. Further research needs to be undertaken on how to model this and complements the existing field of research into preferential host switching \citep{charleston2002preferential,cuthill2013simple}. We, however, do not consider this constraint herein, and assume spread is permissible between all hosts, as a means to present this evolutionary event's value to widespread parasitism analysis.

Formally we define the spread event as a parasite lineage that due to new opportunities infects a new host lineage while maintaining its infection of its current host lineage. This event therefore consists of a shift of a subset of the parasite lineage from the initial host (the take-off edge) to a new host (the landing edge), occurring at some point after the host lineages have diverged. 

This definition is derived from the existing definition of the host switch event with which spread shares a number of common traits. Both events require the internal node ordering of the \emph{host} phylogeny to be fixed, to ensure that the resultant map is time consistent, and both events require that the take-off and landing edges share a common timing interval \citep{conow2010jane}. Spread events, unlike host switch events however, do not consist of a bifurcation, and as a result are not dependent on the internal node ordering of the \emph{parasite} phylogeny, which results in spread being a more generalised version of a host switch event.

With the addition of the spread event we are required to update the cost vector $V$ to include $S$, the cost of a spread event, along with updating the objective function $E$ as follows:

\begin{equation}\label{eq:updatedCostFunction}
E = \alpha C + \beta D + \gamma W + \delta L + \epsilon F + \zeta S
\end{equation}

\noindent where $\zeta$ represents the number of spread events in the resultant map, $\Phi$. It is important to note that even with the addition of spread as an additional evolutionary event, the total number of widespread events in Equation~\eqref{eq:updatedCostFunction} $(\epsilon + \zeta)$ is equal to the number of failure-to-diverge events in Equation~\eqref{eq:costFunction} $(\epsilon)$.

Using this new formulation of the objective function $E$, we derive a polynomially bounded algorithm to solve the cophylogeny reconstruction problem where widespread parasites are permitted (the widespread parasites problem), where the internal nodes in the host phylogeny are fixed. The proposed method extends the Improved Node Mapping algorithm \citep{drinkwater2014improved,drinkwater2015subquadratic}, to recover solutions to the widespread parasite problem using both spread and failure-to-diverge. The described methodology, however, is designed so that it can be integrated into other mapping algorithms which leverage a fixed internal node ordering, such as Edge Mapping \citep{yodpinyaneehmc} and Slicing \citep{doyon2011efficient}. This method is then integrated into an existing metaheuristic framework similar to that implemented in Jane \citep{conow2010jane}, which allows for this method to provide robust estimations for the widespread parasites problem in a reasonable period of time.

\subsection{The order of evolutionary events}

Along with integrating both spread and failure-to-diverge within a common framework, our model aims to provide additional flexibility when inferring the position of a widespread event within the reconciled map. Current state of the art algorithms such as Edge Mapping applied in Jane, provide strict bounds on the position where a failure-to-diverge event may occur. These bounds only allow for a subset of the total number of mapping locations to be considered prior to the widespread event. For example consider the tanglegram in Figure~\ref{fig:caseWhereJaneFails} (left) which includes a single widespread parasite. The minimum cost map inferred by Jane, Figure~\ref{fig:caseWhereJaneFails} (right), for this specific instance includes 2 failure-to-diverge events, 1 host switch event and 1 loss event. 

There is an alternate reconstruction for this system, however, where the minimum cost map contains 2 failure-to-diverge events and 1 codivergence event, Figure~\ref{fig:caseWhereJaneFails} (centre). Under all previously published cost schemes this map is considered more parsimonious. Jane is unable to reconstruct this specific map, however, as its algorithm enforces constraints on the number of locations where divergence events may be placed following a set of widespread events. This bound is appropriate as it does allow for a faster running time, however, this bound in cases such as this may give rise to reconciliations which are less parsimonious.

As computational power continues to become faster and cheaper, it is important to consider alternate algorithms which, while potentially less efficient, may provide more parsimonious solutions to the widespread parasite problem. This is the concept which is explored herein, where our proposed framework permits divergence events to occur at all feasible positions prior to and following a set of widespread events. This will increase the asymptotic complexity relative to Jane, in the hope of providing a more parsimonious reconciliation for the resultant maps.

We will show that by increasing the asymptotic complexity by a factor of $n$ that it is possible to provide a solution to the widespread parasite problem which considers both failure-to-diverge and spread, and provides a significant accuracy improvement which is representative of one of the largest single improvements offered by a coevolutionary analysis technique since \citeauthor{charleston1998jungles} \citeyearpar{charleston1998jungles} proposed the Jungle data structure.

\subsection{Integrating Widespread Events into Improved Node Mapping}

In this section we introduce a series of amendments which when applied to the Improved Node Mapping algorithm allows for both failure-to-diverge and spread events to be recovered optimally when reconciling a pair of phylogenetic trees. Prior implementations of node mapping by \citeauthor{libeskind2009computational} \citeyearpar{libeskind2009computational}, and \citeauthor{drinkwater2014improved} \citeyearpar{drinkwater2014improved,drinkwater2015time} have only considered the case where a parasite may inhabit a single host. The amendment described herein not only updates the Improved Node Mapping algorithm to support widespread events, but also resolves the problems associated with algorithms such as Jane which were discussed in the previous section.

The updated version of the Improved Node Mapping algorithm which we will refer to as WiSPA (WideSpread Parasitism Analyser) can be more easily described as a two step process. The first reconciles all optimal widespread events based on an event costs vector, $V$. This is a reconciliation step which recovers all feasible widespread events, where the second step recovers the optimal set of divergence events using the previously derived set of widespread events. By handling these two complex sets of operations in series it is possible to ensure that a polynomially bound algorithm may be derived for solving the widespread parasite problem, where the internal node ordering of the host phylogeny is fixed. 

Our proposed algorithm reconciles the set of optimal widespread events by constructing a set of widespread association trees, a process which is derived from an earlier method proposed by \citeauthor{page1994parallel} \citeyearpar{page1994parallel}. These association trees are then leveraged to recover the optimal set of widespread events, mirroring much of the work proposed by both \citeauthor{page1994parallel} \citeyearpar{page1994parallel} and \citeauthor{siddall2003brooks} \citeyearpar{siddall2003brooks}. Unlike their previous attempts to solve the widespread parasitism problem which applied a greedy algorithm, our approach applies a dynamic programming algorithm to ensure that all feasible states may be considered, avoiding the potential problems that may arise due to local minima or excluding large subsets of the problem space.

\subsection{Reconstructing Widespread Associations as Trees}

To reconcile the set of widespread events for each widespread parasite $(p_i)$, we propose a method which translates the set of widespread associations for the parasite node $p_i$ to a bifurcating tree. A similar model was first used by \citeauthor{page1994parallel} in \citeyear{page1994parallel} to reconcile the widespread parasitism identified in the pocket gopher chewing lice coevolutionary system introduced by \citeauthor{hafner1988phylogenetic}, \citeyearpar{hafner1988phylogenetic}.

The constructed trees which are referred to herein as Association Trees ($a_i$), are a set of trees $A = (a_1 \dots a_n)$, which may be used to infer the optimal set of widespread events where we prove that:

\newtheorem{myLemma}{Lemma}
\begin{myLemma}
An association tree $(a_i)$ is a bifurcating tree constructed based on the associations, $\varphi$, present for the parasite leaf node $p_i$ which mirrors the topology of $H$, such that $a_i$ may infer the maximum number of widespread events.
\label{lemma:listOfCurrentAndFutureCherries}
\end{myLemma}

\begin{proof}
Consider the parasite leaf node $p_i$ with $k$ widespread associations. The maximum number of possible widespread events is the case where $k$ failure-to-diverge events may be recovered. This is because for all cases it is possible to recover $k$ spread events for all trees, due to the construction of the host tree, such that all leaves share a common timing interval (the present) \citep{conow2010jane}.  Therefore an association tree, $a_i$, which maximises the number of failure-to-diverge events will maximise the total number of possible widespread events.

A mirrored tree constructed in line with \citeauthor{fahrenholz1913ectoparasiten}'s \citeyearpar{fahrenholz1913ectoparasiten} Rule will always permit $k$ failure-to-diverge events, as each internal node in the mirrored tree corresponds to an internal node in the host tree \citep{fahrenholz1913ectoparasiten,paterson2001analytical}. Therefore if we construct $a_i$ for $p_i$ which mirrors $H$ based on the associations $\varphi$ ,then we will maximise the number of possible widespread events which are able to be recovered using the association tree.
\end{proof}

By maximising the number of possible events recovered, we ensure that the optimal set of widespread events may be inferred. This is due to the order of widespread events being unbounded, as widespread events are not dependent on the internal order of $P$. Therefore, this approach while ensuring that an optimal set of widespread events is recovered, does not guarantee that the order of events inferred is correct, as there is no information in the initial problem instance to provide such an inference. Further, information about the problem instance would be required to infer the order of widespread events, such as the geographical history of both host and parasite. This, along with the consideration of preferential spread events, is a topic to be considered in later revisions of the WISPA algorithm.

In order to construct the association trees in line with \citeauthor{fahrenholz1913ectoparasiten}'s \citeyearpar{fahrenholz1913ectoparasiten} Rule, we find the unique subtree where each leaf in the association tree is associated with one of the initial widespread associations. The recovery of an association tree can therefore be reduced to the problem of recovering the homeomorphic subgraph of $H$ for the leaves inhabited by the widespread parasite $p_i$ \citep{lozano2007seeded}. 

To construct the homeomorphic subgraph we apply the pruning algorithm described in detail by \citeauthor{lozano2007seeded} \citeyearpar{lozano2007seeded} which creates a copy of $H$ where only the host leaves inhabited by $p_i$ are retained. This algorithm is applied for each widespread parasite which gives rise to the set $A = (a_1 \dots a_n)$.

The associations trees $A = (a_1 \dots a_n)$ mirror $H$ based on each parasite's widespread associations and therefore each leaf in the association tree $a_i$ has a one-to-one association with a leaf in $H$, such that each leaf node in the association tree $a_i$ maps to a unique leaf node in $H$.  This property is not one that is imposed on a standard tanglegram, but is an important property that we leverage to reconstruct widespread events (see next section).

\subsection{Recovering Widespread Events}

The widespread events considered herein are derived from existing divergence events, and therefore existing techniques for the recovery of divergence events may be applied to their recovery. This approach while used by \citeauthor{page1994parallel} \citeyearpar{page1994parallel} to infer failure-to-diverge event,s has not been applied to reconcile multiple widespread parasites within a single common framework. To achieve this each widespread event is considered as the divergence event which most closely matches its behaviour. Under this constraint a failure-to-diverge is recovered from an association tree as a codivergence, and a spread is recovered from an association tree as a host switch.

This is possible as both the optimal codivergence and failure-to-diverge events occur at the most recent common ancestor of their children \citep{johnson2003parasites}, while the optimal host switch and spread events may be recovered using an implementation of the level ancestor problem \citep{drinkwater2014improved}. This is possible as each widespread event mirrors these two divergence events, with the exception that neither include a divergence. This is resolved by creating pseudo-divergence events through the construction of the association trees in line with \citeauthor{siddall2003brooks}'s, \citeyearpar{siddall2003brooks}, proposed reconciliation model.  

Therefore as both widespread events can be inferred from existing divergence events, we may apply existing solutions to the dated tree reconciliation problem, as a means to recover the  optimal divergence events for the set of association trees, $A$. This may in-turn be leveraged to infer the optimal set of widespread events for each association tree, $a_i$. This is possible as association trees are constructed with a one-to-one mapping, which mitigates the need for duplication events, if host switch events are permitted. This is important as there is no widespread equivalent for a duplication event. Exploiting this imposed property of each association tree, we can reconstruct the map for each association tree where only codivergence, host switch and loss events are permitted; that is running the existing Improved Node Mapping algorithm with a cost vector of $(F, \infty, S, L)$, where the costs for failure-to-diverge $(F)$ and spread $(S)$ replace the costs for codivergence and host switch respectively.

The widespread events are inferred from the recovered mappings by relabelling each codivergence as a failure-to-diverge and each host switch as a spread. This process requires that each divergence event in the resultant dynamic programming table generated by the Improved Node Mapping algorithm may be replaced with its corresponding widespread event. The inferred widespread events are then retained within a dynamic programming table $d_i$, which contains all the optimal widespread events for the parasite $p_i$. Therefore the result of mapping the complete set of association trees $A$ into $H$ gives rise to a set of dynamic programming tables $\omega = (d_1, \dots d_n)$, containing all the optimal widespread events for the parasite tree $P$. 

The ReconcileWidespreadParasite algorithm applied to infer the complete set of optimal widespread events for a parasite tree $P$ with respect to its host $H$ is defined in Figure~\ref{fig:reconcile}. This process outlines a new approach to reconciling the incongruence caused by widespread parasitism. It integrates a number of existing approaches proposed by \citeauthor{page1994parallel} \citeyearpar{page1994parallel}, \citeauthor{siddall2003brooks} \citeyearpar{siddall2003brooks}, and \citeauthor{brooks1991phylogeny} \citeyearpar{brooks1991phylogeny}, along with integrating the works of \citeauthor{banks2005multi} \citeyearpar{banks2005multi}, and \citeauthor{johnson2003parasites} \citeyearpar{johnson2003parasites} into a single reconciliation methodology within the context of dated trees. This in turn provides the foundations to infer the optimal set of divergence events, which is described in detail in the following section.

\subsection{Recovering Divergence Events}

The recovery of the divergence events using WiSPA is derived from traditional bottom-up (taxa-to-root) dynamic programming approaches applied in the Slicing \citep{doyon2011efficient}, Edge Mapping \citep{yodpinyaneehmc}, and Improved Node Mapping \citep{drinkwater2014improved} algorithms. Each of these existing approaches incrementally constructs their resultant map using a series of sub-solutions, leading to the recovery of an optimal mapping of the parasite phylogeny into its host. 

One such method, the Improved Node Mapping algorithm, is a cubic time solution for the dated tree reconciliation problem. This approach reconciles the incongruence displayed for each parasite node, by reconciling the optimal divergence event based on the set of mapping sites for its children. This requires a nested set of loops so that every mapping site for the left child is compared with every mapping site of the right child. 

The WiSPA algorithm unlike Improved Node Mapping considers multiple optimal locations for each parasite node, rather than a single optimal mapping site which has been the premise of all cubic time solutions to this problem \citep{doyon2011efficient,yodpinyaneehmc,drinkwater2014improved}. In this more complex case, rather than an optimal mapping site for each pair of children, the optimal mapping may occur at any location, with the sub-solution defined by the widespread mapping. Initial analysis may suggests that this additional complexity may induce a further set of quadratic comparisons. This, however, can be mitigated by exploiting a number of topological properties of the underlying dynamic programming table and the topology of the resultant map, both of which are explored within this section.

The first point which should be noted is that the dynamic programming table traditionally only retains a single mapping site for the parasite leaves \citep{drinkwater2015subquadratic}. It is possible, however, to retain multiple mappings for each parasite leaf node, where in fact there is the ability to retain a mapping site for each parasite node to all locations in the host tree, without increasing the asymptotic complexity of the Improved Node Mapping algorithm. Exploiting this property of the dynamic programming table in handling widespread parasites was introduced as a possibility during the formulation of the original Improved Node Mapping algorithm \citep{drinkwater2014improved}. By allowing a set of mappings for each parasite node of this size, allows for the optimal set of mapping sites stored for the root of the association tree $a_i$, corresponding to the parasite node $p_i$ in question, to be retained within the dynamic programming table. This in turn allows for the optimal set of widespread events to be considered within the context of inferring a set of optimal divergence events.  

To handle the additional complexity which arises due to handling multiple widespread parasite events, the Improved Node Mapping algorithm has been split such that it considers three possible scenarios, including the case where the left child is treated as a widespread parasite, the right child us treated as a widespread parasite or neither child is treated as a widespread parasite, as can be seen in Figure~\ref{fig:wISPA}. In the case where the left or right child is treated as a widespread event (lines 17 - 28), the divergence event may be placed at an earlier time period to root of the widespread event (either a failure-to-diverge or a spread event), as long as the relative order of the parasite phylogeny is preserved. That is while a divergence event may be placed prior to multiple widespread events, it may never be placed at a position prior to one of its descendants. Prior in this context refers to a position closer to the present, as the solutions are constructed in reverse, from the tips to the root. To provide this additional degree of flexibility when reconciling the incongruence between the parasite and its host, requires that all positions within the host be considered as a possible mapping site for each pair of points, adding an additional nested loop (on lines 18-21 and 25-27, in the case where the left or right child are widespread respectively). 

Handling widespread parasitism in this fashion results in either a widespread event being the root of a sub-solution, such as a failure-to-diverge event occurring at a time period in the past before any of the divergence events, or a divergence event occurring as the root of a sub-solution. In the latter case this sub-solution from this point onwards is considered as a standard mapping site, in line with previous models, while in the case where the root is a widespread event, its parent too will be required to traverse the complete search space to allocate the optimal divergence event, and therefore an additional layer of computational complexity is added with this approach, discussed in detail in the following section.

The major benefit of this model is that in the case where both the left and right children are widespread parasites, it is possible to abstract away any possible compounding complexity by considering each widespread parasite in series. This reduces the need for an additional increase in the computational complexity of the proposed model, which is achieved by noting that a divergence event may not occur prior to the root of both widespread events, as this would reflect the occurrence of divergence events, and as such one of the two widespread events must be considered as a root, or the divergence event itself may be the root of both lineages. This is in line with the theory considered by \citeauthor{fish2013cophylogeny}, \citeyearpar{fish2013cophylogeny}, in the development of the third version of Jane was the first version to consider widespread parasitism.

In the final case (lines 29 - 32) neither the left or right child are rooted by widespread events. In this case the complexity of widespread parasitism is already fully explained within the sub-solution, or the sub-solution does not contain any widespread events. In either case such a sub-solution is processed in-line with the existing Improved Node Mapping algorithm, and no further changes are required to the algorithm presented in Figure~\ref{fig:wISPA} to handle this case.

Therefore by reconciling the optimal set of divergence events based on the optimal set of widespread evolutionary events retained within $\omega$ it is possible to handle multiple widespread evolutionary events and to overcome the limitations identified within the algorithm applied by Jane. In the following section the asymptotic complexity of the algorithm is discussed, where we prove that the additional accuracy provided by the model is achieved by adding only a $O(n)$ increase in the complexity of the Improved Node Mapping algorithm, resulting in a complexity which is comparable to software tools such as Costscape and Eventscape \citep{libeskind2014pareto} and significantly faster than Jane 1 \citep{conow2010jane} and the Jungle method \citep{charleston1998jungles,TreeMapWebsite}, all of which are popular co-evolutionary analysis methods. 

\subsection{Complexity Analysis}

The WiSPA algorithm is designed using a series of underlying algorithms to provide the most accurate algorithm for handling widespread parasites. In this section we analyse the associated computational complexity of this approach, and how this compares to existing algorithms applied within the field of coevolutionary analysis of widespread parasites.

For the complexity analysis considered herein we consider the number of nodes in the host tree to be $2n-1$. That is that the host tree contains $n$ leaves and $n-1$ internal nodes. The parasite tree conversely contains $2m-1$ nodes, where the parasite tree contains $m$ leaves and $m-1$ internal nodes. Finally the maximum number of associations for an individual widespread parasite is considered as $k$ where $k \leq n$. That is no single parasite may have more associations then there are unique host leaves to infect.

The WiSPA algorithm is composed of two computationally expensive steps. The first is the processing required to handle the parasites which inhabit more than one host, specifically constructing and solving the association trees (lines 7-12 in Figure~\ref{fig:wISPA}), and the second step is processing the divergence events, the internal nodes in the parasite tree (lines 14-34 in Figure~\ref{fig:wISPA}).

Processing the leaves in the parasite tree requires the construction of $O(m)$ association trees which are of size $O(k)$. The association trees are constructed using an application of \citeauthor{lozano2007seeded}'s \citeyearpar{lozano2007seeded} homeomorphic subgraph pruning algorithm, which runs in $O(kn)$ for each of the $O(m)$ widespread parasites. Therefore the time required to construct the set of association trees, $A$, is $O(kmn)$. The solutions for each of these association trees are stored with an array of dynamic programming tables, where each table is of size $O(nk)$, where the array of dynamic programming tables $\omega$ contains $O(m)$ elements. Therefore the space requirement for the step is $O(kmn)$. Solving each of the association trees requires $O(kn^2)$ time, and therefore as $O(m)$ trees need to be solved the total running time of this step is $O(kmn^2)$.

Reconciling the divergence events (lines 14-34 in Figure~\ref{fig:wISPA}) requires mapping the parasite into the host using the additional information retained within the list of dynamic programming tables, $\omega$. As the additional widespread information is retained within $\omega$ no additional space is required compared to the original dynamic programming table construction defined by \citeauthor{drinkwater2014improved} \citeyearpar{drinkwater2014improved}, and therefore the space required is $O(mn)$. The running time however requires an additional step which involves iterating over all the possible widespread locations of which there may be $O(k)$ for each mapping site considered, and therefore the running time is extended from $O(mn^2)$, as defined within the original implementation of the Improved Node Mapping algorithm, to $O(kmn^2)$.

This time and space complexity is quite significant considering that the complexity of the proposed algorithm grows linearly in regards to the number of additional widespread associations which are added to the tanglegram. That is while the Improved Node Mapping algorithm runs in cubic time when considering only $O(n)$ associations, our proposed algorithm runs in quartic time when considering $O(n^2)$ associations. Therefore in the case where only one additional widespread association is added to each parasite, the total running time only increases by a factor of two. This is significant as the number of widespread associations for each parasite will never be of size $O(n)$ under any realistic biological scenario. For example, if we consider the 15 previously published biological data sets introduced later to validate our model, it may be observed that on average the rate of widespread parasitism is approximately 7\%, which compared to the size of the data sets is less than $\log{n}$, which argues that while the worst case running time for the proposed algorithm is quartic, the actual running time in practice is actually more comparable to existing cubic time algorithms.

\subsection{Implementation and Validation}

The algorithm proposed herein is implemented in Java and is available as a platform-independent jar file. The underlying algorithm is integrated into a genetic algorithm, which is designed to run in a multithreaded environment, similar to the design proposed by \citeauthor{conow2010jane} \citeyearpar{conow2010jane}. The advantage of \citeauthor{conow2010jane}'s \citeyearpar{conow2010jane} model is the near-linear speedup possible using multi-core systems.

Jane 4 \citep{conow2010jane} was selected as the algorithm to validate the theoretical model presented herein. Jane is the best candidate to evaluate the performance of WiSPA as both methods are designed to minimise the total cost of all evolutionary events considered, and that they both leverage an underlying algorithm to solve the dated tree reconciliation problem as a means to inform their metaheuristic framework. CoRe-PA and CoRe-ILP were not considered, as for the size of the data sets considered herein Jane has been shown to outperform both these techniques \citep{conow2010jane,wiesekecophylogenetic}. 

The evaluation of our new model is broken into two parts. The first considers Jane and WiSPA's accuracy over 500 synthetic data sets which display varying degrees of widespread parasitism. Then Jane and WiSPA are evaluated over 15 previously published biological systems. In both evaluations two key metrics are considered. The first is the total cost of the reconciliation inferred by each model, and the second is the total number of codivergence events present in the inferred reconciliation. Each of these two values represent the degree of congruence represented by the reconciled map, where the aim for coevolutionary analysis is to infer the minimum cost map with the maximum number of codivergence events \citep{littlewood2003evolution}. Therefore each model will be validated on how well they conform to this criteria. These two key metrics align with prior analysis of coevolutionary techniques \citep{page1994maps,page2002tangled,ronquist1998three,conow2010jane,wiesekecophylogenetic}, and are considered the best two signals for recovering a biologically relevant map which most accurately represents the actual coevolutionary interactions.

Along with demonstrating the effectiveness of the generalised model applied within WiSPA, this analysis also aims to infer the significance of the inclusion of the spread evolutionary event. This was achieved by considering three different costs for the evolutionary event spread; a cost of one, which is equal to the cost of a failure-to-diverge event, a cost of two, the same cost as a host switch event the evolutionary event which is most similar to the spread event, and finally the case where a spread event is not permitted (in essence assigned a cost of $\infty$).

Each of these values for the spread event are integrated into the Jungle cost scheme \citep{ronquist2003parsimony} to provide three unique event cost schemes for this analysis, including $V=(0,1,2,1,1,1)$, $V=(0,1,2,1,1,2)$, and $V=(0,1,2,1,1,\infty)$. When evaluating the performance of Jane using these cost vectors the recovered map will always have the same cost, as varying the cost of spread has no bearing on the cost of the recovered map by Jane.

\section{Discussion and Analysis}

The analysis performed herein using a combination of synthetic and biological data sets will demonstrate that WiSPA is able to converge on maps with a lower event cost, with a significantly higher number of codivergence events. The significance of this result is that even in the case where spread is not permitted, WiSPA is observed to perform 2\% better in practice. This is shown to only improve as spread is permitted, and its associated penalty cost is reduced.

Following this successful result we continue our analysis of WiSPA and Jane by considering the \emph{Primate}--\emph{Enterobius} biological data sets in further detail. This biological system has long been considered a likely coevolutionary system, however, the inability of prior models to handle the widespread parasitism resulted in no previous model providing a statistically significant coevolutionary hypothesis for the sub-clade considered herein. We demonstrate that while Jane may be unable to provide such a hypothesis, for certain values of spread, WiSPA is able to provide a statistically significant coevolutionary hypothesis for this system.

\subsection{Overall Performance on Synthetic Data}

The synthetic data sets used to evaluate WISPA were previously constructed using the Cophylogeny Generation Model (Core-Gen) \citep{keller2011evaluation}. These coevolutionary histories were constructed using a standard Yule Model, a common synthetic tree generation model applied in phylogenetics \citep{steel2001properties}. Previously \citeauthor{keller2011evaluation} \citeyearpar{keller2011evaluation} constructed 1000 synthetic data sets, where for this evaluation we have randomly selected 50 of these to provide a baseline for this comparison.

As Core--Gen can only generate coevolutionary systems where each parasite infects a single host, the existing data sets needed to be modified to induce widespread parasitism. From each of the 50 synthetic data sets initially selected, nine additional data sets were created by randomly applying additional widespread associations to the initial data sets. These additional nine new data sets present a varied degree of widespread parasitism, with the aim to model a decreasing rate of host specificity.

For the nine data sets we allowed additional widespread events to be added for each parasite, such that the maximum rate of additional widespread parasitism was 10\%, 15\%, 20\%, 25\%, 30\%, 35\%, 40\%, 45\%, and 50\% of the total available host species for each of the nine data sets respectively. This is applied by selecting each parasite node and allowing the parasite to infect a random number of additional host species between $(0$ and $p \times n)$, where $p$ is the rate of widespread parasitism for the specified synthetic system and $n$ is the number of host taxa. It should be noted that this model is a crude representation of widespread parasitism in nature, however, it provides a robust set of synthetic data sets to compare Jane and WiSPA over varying degrees of widespread parasitism.

This technique is also advantageous as it provides a baseline of the number of codivergence events present in the original tanglegram, which can be compared to the recovered number of codivergence events of each technique as the rate of widespread parasitism increases. Therefore both the rate at which the total event cost increases and the total number of codivergence events decreases, may be tracked for the models considered within this analysis.

This is captured in Figure~\ref{fig:resultsForMetaheuristicPlots} where the total event cost (left) and total number of codivergence events (right) are recorded for the ten data sets (including the baseline) for the four models considered. These two plots provide the best insight to date in regards to the benefits of the spread event, particularly in increasing the total number of codivergence events compared to using failure-to-diverge exclusively.

In the case where spread is set to one a reduction of more than 50\% is achieved in the parsimony score, with this reduction only increasing as the rate of widespread parasitism increases. This reduction is complimented by a nine fold increase in the number of codivergence events. A similar trend is observed where spread is set to two. Here a reduction of more than 35\% is achieved in the parsimony score, with this reduction only increasing as the rate of widespread parasitism increases. The reduction in this case is complimented by an eight fold increase in the number of codivergence events. 

In both cases where spread is permitted a significant improvement in the congruence of the reconciled maps is achieved. In the case where spread is not permitted such a drastic improvement is not observed although there is a 1\% decrease in the total parsimony cost compared to Jane with an increase of 2\% in the number of codivergence events. While nowhere near as impressive as the case where spread is permitted, this improvement is important as while it represents our model in the worst case it still shows it outperforms Jane and only improves as the cost of spread is decreased. 

\subsection{Overall Performance on Real Data}

The performance over the synthetic data set demonstrates the value of the generalised model applied within WiSPA along with the advantage of applying the spread event for the analysis of systems presenting widespread parasitism. As noted, however, the model applied to generate the synthetic data sets does not provide the best representation of widespread parasitism within a biological system, although it is the first synthetic model which attempts to model this phenomenon. Therefore our analysis also compares the performance of Jane and WiSPA over biological data sets.

For this analysis 15 biological data sets were selected to compare WiSPA with the latest version of Jane. These biological systems considered 10 biological phenomena, including but not limited to parasitism \citep{hafner1988phylogenetic}, plant--insect interactions \citep{gomez2010neotropical}, coevolutionary dynamics between a virus and its host \citep{jackson2004cophylogenetic}, mutualism \citep{mcleish2012codivergence}, parasitoidism \citep{murray2013ancient}, and plant--fungal coevolution \citep{refregier2008cophylogeny}. The complete list of biological systems included in this analysis and the coevolutionary interrelationships each system expresses has been listed in Table~\ref{tab:realDataSets}. These data sets were selected to evaluate whether spread can assist in recovering more parsimonious reconstructions, along with evaluating the newly proposed model for reconciling widespread parasitism. This analysis along with evaluating the cost of each reconciliation, also considers the number of codivergence events recovered as another means of evaluating the inferred congruence found by each technique considered herein. 

It should be noted that while 15 data sets does not compare to the 500 data sets considered in the previous section, this selection of biological data represents the largest collection of coevolutionary systems displaying widespread parasitism assembled to date. While larger collections exist for the case where parasite only parasites infect a single host such as the 102 data sets catalogued by \citeauthor{drinkwater2014treeCollapse} \citeyearpar{drinkwater2014treeCollapse} this is the first and therefore largest selection of widespread coevolutionary systems aggregated to date.

The results of this comparison are displayed in Table~\ref{tab:realDataResults} and show a significant improvement in both the reconciliation's cost, and the total number of codivergence events when including the spread event in the reconstruction of the parasites' evolutionary history with respect to their host. In all cases where spread was assigned a cost of one, the newly proposed algorithm found a solution that was at least as parsimonious as Jane, with the majority of cases inferring a solution which was significantly more parsimonious and with a higher number of codivergence events. Over the 15 data sets there was an observed reduction of 35\% in the event cost, and an increase of 21\% in the number of codivergence events.

A similar trend was observed in the case where spread was assigned a cost of 2. Here in all cases WiSPA was able to find a solution which was at least as parsimonious in terms of event cost, with a number of cases where WiSPA was able to infer a reconciliation which was significantly more parsimonious and with a higher number of codivergence events. Over the 15 data sets there was an observed reduction of 22\% in the event cost and an increase of 18\% in the number of codivergence events.

This demonstrates the value of the spread event for coevolutionary analysis, where a significant reduction in the total parsimony cost may be achieved using the spread event, with the largest single reduction providing a 55\% decrease in the total event cost. These results match the benefits observed over the synthetic data sets, providing further evidence of the value of adopting the spread event for widespread analysis.

In the case where spread was not permitted, WiSPA was still able to outperform Jane, although the improvement was not as pronounced. Overall there was a 2\% reduction in the total cost of the 15 maps with no difference in the total number of codivergence events inferred across the 15 systems. This is still a significant result as this is the worst case performance of the newly proposed model for reconciling widespread parasitism, and even then we are able to present an improvement of 2\%. While in the case where spread is not permitted many systems perform as well as Jane, there is one particular system which displays a significant improvement. The ant--wasp parasitoid coevolutionary system's cost is reduced by  25\% by allowing a divergence event to occur prior to multiple widespread events. Such a significant reduction is the difference between a map which is only cheaper than 61.19\% of random solutions as in the case of Jane compared to a map which is cheaper than 94.96\% of random solutions as in the case of WiSPA. These results are based on the randomisation test undertaken using Jane, using 10000 random instances.

While in all cases WiSPA was able to recover a map which was equal to or less than that which was recovered by Jane it should be noted that there was a case where Jane was able to outperform WiSPA in terms of inferring a map where the total number of codivergence events was higher. In the RNA Virus example, which has been marked as bold in Table~\ref{tab:realDataResults}, it can be seen that Jane's best reconciliation contains 5 codivergence events while WiSPA only infers a map with 4. In both cases the recovered map has a cost of 15 and as such current significance testing considers these two model equivalent. If subjected to the same randomisation test considered for the ant--wasp parasitoid coevolutionary systems neither map is considered significant, $(p=0.14)$.

The results over the biological data sets provide a compelling case for the adoption of the newly proposed model. Here a reduction in the parsimony score of 22\% is achieved even where spread is set to a cost of two, reducing a further 14\% if spread is assigned a cost of one. In the following section we demonstrate the significance of this improvement by comparing Jane's reconciliation with those of our model along over a long studied biological system of public health significance.

\subsection{Spread provides stronger evidence for \emph{Primate}--\emph{Enterobius} Coevolution}

\emph{Primate} and \emph{Enterobius} (Pinworms) have long been considered as a possible coevolutionary system \citep{cameron1929species,sandosham1950enterobius,sorci1997host,hugot1999primates}. This hypothesis is due to the high degree of congruence which has been observed between these two phylogenetic trees \citep{hugot1999primates}. \citeauthor{brooks1982pinworms} \citeyearpar{brooks1982pinworms} identified that while the observed congruence within this system strongly supported coevolution that there remains a subset of the \emph{Primate} and \emph{Enterobius} tanglegram which did not appear to provide evidence of coevolution. In particular it was noted that the species \emph{E. vermicularies} infection of both \emph{Hylobatidae} (Gibbon) and \emph{Homo sapien} (Human) could not be explained by traditional coevolutionary models.

This failure by cladistic models was due to the inability of coevolutionary analysis to reconcile widespread parasitism. This was later rectified as part of SBPA which \citeauthor{brooks2003extending} \citeyear{brooks2003extending} applied to this system although once again no specific modelling for the relationships between the species \emph{E. vermicularies} which inhabits both humans and gibbons was provided. This unexplained sub-clade has also been considered by cophylogeny models as well as cladistic approaches, with \citeauthor{ronquist1997phylogenetic} \citeyearpar{ronquist1997phylogenetic} proposing that this inconsistency was due to a recent host switch event from gibbons to humans. One weakness with this hypothesis, however, is that it assumes that \emph{E. vermicularies} has diverged during the infection of Humans which current evidence does not support \citep{brooks1982pinworms}. As a result, a complete hypothesis which reconciles the observed data within this potential coevolutionary system and in particular this sub-clade has remained unanswered.

While initial coevolutionary analysis assumed widespread parasitism cannot occur \citep{poulin2011evolutionary} in a coevolutionary context, this has gradually become more accepted as potentially occurring depending on the nature of the coevolutionary system considered. Laboratory experiments have shown that \emph{Enterobius} is a species which displays a low host specificity, with results as early as \citeauthor{sandosham1950enterobius} \citeyearpar{sandosham1950enterobius} noting that a number of species of \emph{Enterobius} infected phylogenetically distant primates held within captivity and which would not associate with one another in the wild due to vast geographical diversity. From this evidence it does not seem infeasible that \emph{E. vermicularies} may also be able to infect multiple host species.

The infection of humans has a higher probability due to humans no longer being bound by their biogeographical environment. This hypothesis agrees with existing coevolutionary analysis focusing on tapeworms, which have been shown to have a low host specificity wereby species were able to infect humans during their dispersal from Africa 2.5 million years ago \citep{hoberg2001out}. 

Therefore, we attempt to provide a coevolutionary explanation applying widespread parasitism to this sub-clade using both the methodologies applied in Jane and WiSPA. We firstly evaluate the two recovered maps from Jane and WiSPA and discuss their inferred set of biological events and their implications. These maps are then evaluated statistically to evaluate if either method rejects the null hypothesis that these two phylogenetic trees are independent from one another. For this analysis we apply the Jungle cost scheme \citep{ronquist2003parsimony} including spread with a cost of both one and two.

To provide a fair statistical analysis we generate all feasible widespread systems which include one additional widespread association between the parasite and its host. In total there are 10000 systems where the host and parasite phylogenies are fixed and the associations are randomised. By generating a single instance of all possible maps, we guarantee that no bias is introduced using different randomisation techniques for each model. 

These models may be generated by computing all possible association pairs of which there are $5^4$ and multiplying this by the total number of unique additional associations that may be applied, which is $(5-1) \times 4$. The minus one is due to the inability to apply more than one association between a single parasite and a single host. Therefore the total number of unique systems that may be generated for the \emph{Primate} and \emph{Enterobius} (Pinworms) system presented in Figure~\ref{fig:primatePinwormsTanglegram} is $10000$ $(5^4 \times(5-1) \times 4)$.

Jane's recovered map for the Jungle cost scheme is visualised in Figure~\ref{fig:primatePinwormMaps} (left). This map consists of two codivergences, one host switch, three loss events and one failure-to-diverge which has a resultant cost of six. This map hypothesises that the species of pinworm which now infects gibbon and human had the potential of infecting all species which humans share a more recent ancestor with than with gibbons but that this species failed to do so on all occasions. This seems highly improbable as gibbon and human's most recent common ancestor is estimated to have lived over 14 million years ago \citep{carbone2014gibbon}.

Further, while the initial parasite phylogeny appears to have a high degree of congruence with its host, this apparent congruence is not well represented in the recovered map. If we compare Jane's recovered map to the inferred cost over the 10000 unique instances, we observe (in Figure~\ref{fig:resultsForBernoulliTrails} (left)) that it is relatively high compared to what would be expected simply by chance. If we evaluate its cost using the Wilson score interval we converge on a confidence interval of $(0.760, 0.795)$. As such the reconciliation argues that the evaluation history of these two sets of species are independent. 

WiSPA's recovered map for the cost scheme $V=(0,1,1,2,1,1)$ and $V=(0,1,1,2,1,2)$ is visualised in Figure~\ref{fig:primatePinwormMaps} (right). In both cases the same map was inferred where the only difference was that the recovered spread event costs more in the latter case. This map consists of three codivergences, one loss event and one spread which has a resultant cost of two or three respectively. This map provides the hypothesis that this system has been coevolving throughout its evolutionary history with all divergence events indicative of coevolution.  This widespread event for \emph{E. vermicularies} in this map is explained using a recent spread event from gibbon to human. As previously discussed spread requires that the hosts are biologically collocated at the time spread occurs. This collocation can be explained in this case as humans are no longer geographically bound and therefore spread's potential is significantly higher than for other geographically bound primates. The loss event in this reconstruction can be explained by integrating citeauthor{ronquist1997phylogenetic}'s \citeyearpar{ronquist1997phylogenetic} prior hypothesis. In particular he noted the introduction of \emph{E. vermicularies} into humans may have caused an extinction of a species of \emph{Enterobius} with a common ancestor of \emph{E. anthropopitheci}. 

If we compare WiSPA's recovered map to the same 10000 unique instances that were used to evaluate Jane's map, it can be seen that there is strong evidence for coevolution in this case, as seen in Figure~\ref{fig:resultsForBernoulliTrails} (center and right). Using the Wilson score interval we converge on a confidence interval for the case where spread is a cost of one and two of $(0.028, 0.045)$ and $(0.044,0.069)$ respectively. In both cases this presents a strong indication of coevolution, considering the size of the tanglegram where even a perfectly congruent tanglegram with four internal nodes in the host and parasite tree (an additional node in the parasite compared to this example) only offers a confidence value of $(0.011, 0.023)$.

Due to these results we argue that the inconsistent sub-clade from the \emph{Primate} / \emph{Enterobius} tanglegram can be explained using widespread parasitism when applying the spread event. In particular we note that the algorithm WiSPA is the only method that is able to recover a widespread solution to this instance and provide a statistically significant signal for coevolution for this evolutionary system. 

\section{Conclusion}

This work presents a new model for reconciling the incongruence that may arise between a pair of phylogenetic trees where parasites are permitted to inhabit more than one host. While this permutation of the cophylogeny reconstruction problem has often been considered to be computationally complex, we provide a polynomial solution in the case where their exists timing information for the host phylogeny. In the case where such timing information is unavailable, we provide a metaheuristic framework which applies our underlying algorithm, which is shown to be the most accurate model for widespread parasitism produced to date.

The accuracy improvement present within our proposed model (WiSPA) is due to its inclusion of an additional widespread evolutionary event, which we refer to as spread, along with it providing a more generalised framework for inferring the optimal set of widespread events. The additional widespread evolutionary event applied herein is derived from a number of previous coevolutionary models, along with observed parasitic behaviour in nature and the laboratory, where the inclusion of the spread event alone has been shown to provide an accuracy improvement of over 55\%. 

The accuracy improvement comes at a cost however, where our model is shown to be an order of magnitude slower then the current state of the art algorithm applied in Jane. While this algorithm is more computationally expensive than algorithms applied in the latest version of Jane \citep{janeWebsite} and CoRe-PA \citep{merkle2010parameter}, our algorithm is still far superior to Jane 1 \citep{conow2010jane} and the Jungle model applied in TreeMap, which have both been applied to successfully analyse a number of coevolutionary systems, and is also asymptotically more effiecent than the tools within the Xscape framework \citep{libeskind2014pareto}, proving that our model is capable of analysing biological data sets.

Finally we applied WiSPA to the well-studied sub-clade of the coevolutionary system of \emph{Primate} and \emph{Enterobius}. Since this sub-clade was identified by \citeauthor{brooks1982pinworms} \citeyearpar{brooks1982pinworms} no satisfactory explanation reconciliation of this sub-clade has been derived. We have shown that while this has eluded prior models, WiSPA is able to provide a statistically significant hypothesis for this sub-clade which complements the existing theory of \emph{Primate} / \emph{Enterobius} coevolution, and also provides a plausible biological model consistent with broader understanding of primate--parasite coevolution. 

This result coupled with the results when comparing WiSPA and Jane over the synthetic and biological data sets considered herein demonstrates the value of our proposed model. Not only does WiSPA provide the flexibility of providing an additional evolutionary event to explain the incongruence caused by widespread taxa but this model also provides further flexibility in reconciling the conflict that may arise when dealing with the order of widespread and divergence events. As such we argue for the adoption of this new model to provide additional insights into the complex problem of reconciling the coevolutionary associations of widespread taxa.
 

\newpage 
 
\bibliographystyle{sysbio}
\bibliography{References}

\newpage

%%%%%%%%%%%%%%%%%%%%%%%%%%%%%%%%%%%%%%%%%%%%%%%%%%%%%%%%%%%%%
%%                       Tables                            %%
%%                                                         %%              
%%%%%%%%%%%%%%%%%%%%%%%%%%%%%%%%%%%%%%%%%%%%%%%%%%%%%%%%%%%%%
\begin{table}[h]
	\caption{Biological systems considered in this analysis and the type of coevolutionary interrelationship expressed within said system.}
	\begin{center}
		\begin{tabular}{ l L{5.2cm} }
		\hline
		Coevolutionary system & Type of coevolutionary interrelationship expressed \\
		\hline
		\emph{Acacia} / \emph{Pseudomyrmex} \citep{gomez2010neotropical} & Plant--Insect Mutualism \\
		\emph{Aves} / \emph{Syringophilopsis} \citep{hendricks2013cophylogeny} & Bird--Mites Parasitism \\
		\emph{Carex} / \emph{Anthracoidea} \citep{escudero2015phylogenetic} & Plant--Fungi Parasitism \\
		\emph{Caryophyllaceae} / \emph{Microbotryum} \citep{refregier2008cophylogeny} & Plant--Fungi Mutualism \\
		\emph{Cichlidae} / \emph{Platyhelminthes} \citep{mendlova2012monogeneans} & Fish--Flatworm Parasitism \\
		\emph{Formicidae} / \emph{Eucharitidae} \citep{murray2013ancient} & Ant--Wasp Parasitoidism \\
		\emph{Goodeinae} / \emph{Margotrema} \citep{martinez2014historical} & Fish--Flatworm Parasitism \\
		\emph{Ficus} / \emph{Agaonidae} \citep{mcleish2012codivergence} & Plant--Insect Mutualism \\
		\emph{Gastropoda} / \emph{Schistosome} \citep{lockyer2003phylogeny} & Snails--Flatworm Parasitism \\		
		\emph{Geomyidae} / \emph{Mallophaga} \citep{hafner1988phylogenetic} & Rodent--Lice Parasitism \\
		\emph{Mycocepurus smithii} / \emph{Fungi} \citep{kellner2013co} & Ant--Fungal Mutualism \\
		\emph{Ramphastidae} / \emph{Mallophaga} \citep{weckstein2004biogeography} & Bird--Lice Parasitism \\
		\emph{Sigmodontinae} / \emph{Arenaviridae} \citep{jackson2004cophylogenetic} & Rodent--Viral Coevolution \\
		\emph{Teleostei} / \emph{Copepods} \citep{paterson1999have} & Fish--Crustacean Parasitism \\
		\emph{Tephritidae} / \emph{Bacteria} \citep{viale2015pattern} & Fly--Bacteria Symbiosis \\
		\end{tabular}
	\end{center}
	\label{tab:realDataSets}	
\end{table}

\newpage

\begin{table}[h]
	\caption{WiSPA's performance against Jane 4 over fifteen biological test cases. WiSPA has been run with three different costs associated for spread. Spread was set to a cost of 1, 2 and where spread was not permitted in the reconstruction.}
	\begin{center}
		\begin{tabular}{ l r r r r }
		\hline
		Coevolutionary system & \multicolumn{4}{c}{Recovered event cost and (\# of codivergence events)} \\
		 & Jane & Spread = 1 & Spread = 2 & No Spread \\
		\hline
		\emph{Acacia} / \emph{Pseudomyrmex} & 67 (0) & 28 (2) & 43 (2) & 65 (0) \\
		\emph{Aves} / \emph{Syringophilopsis} & 17 (9) & 17 (9) & 17 (9) & 17 (9) \\
		\emph{Carex} / \emph{Anthracoidea} & 73 (9) & 59 (9) & 65(10) & 73 (9) \\
		\emph{Caryophyllaceae} / \emph{Microbotryum} & 33 (3) & 26 (5) & 30 (3) & 33 (3) \\	
		\emph{Cichlidae} / \emph{Platyhelminthes} & 40 (7) & 34 (9) & 39 (7) & 39 (7) \\	
		\emph{Formicidae} / \emph{Eucharitidae} & 12 (0) & 8 (1) & 9 (1) & 9 (1) \\
		\emph{Goodeinae} / \emph{Margotrema} & 36 (2) & 21 (4) & 25 (4) & 33 (2) \\	
		\emph{Ficus} / \emph{Agaonidae} & 10 (3) & 8 (4) & 9 (4) & 10 (3) \\	
		\emph{Gastropoda} / \emph{Schistosome} & 122 (1) & 54 (3) & 77 (2) & 120 (1) \\		
		\emph{Geomyidae} / \emph{Mallophaga} & 9 (6) & 8 (6) & 9 (6) & 9 (6) \\	
		\emph{Mycocepurus smithii} / \emph{Fungi} & 42 (1) & 21 (2) & 28 (3) & 41 (1) \\
		\emph{Ramphastidae} / \emph{Mallophaga} & 17 (2) & 12 (2) & 14 (3) & 17 (2) \\	
		\textbf{\emph{Sigmodontinae} / \emph{Arenaviridae}} & \textbf{15 (5)} & \textbf{15 (4)} & \textbf{15 (4)} & \textbf{15 (4)} \\	
		\emph{Teleostei} / \emph{Copepods}& 4 (1) & 2 (2) & 3 (2) & 4 (1) \\
		\emph{Tephritidae} / \emph{Bacteria} & 29 (12) & 29 (12) & 29 (12) & 29 (12) \\		
		Total & 526 (61) & 342 (74) & 412 (72) & 512 (61) \\
		\end{tabular}
	\end{center}
	\label{tab:realDataResults}	
\end{table}

\newpage

%%%%%%%%%%%%%%%%%%%%%%%%%%%%%%%%%%%%%%%%%%%%%%%%%%%%%%%%%%%%%
%%                       Figures                           %%
%%                                                         %%              
%%%%%%%%%%%%%%%%%%%%%%%%%%%%%%%%%%%%%%%%%%%%%%%%%%%%%%%%%%%%%

\clearpage

\begin{figure}[H]
\centering

\includegraphics[width=\textwidth]{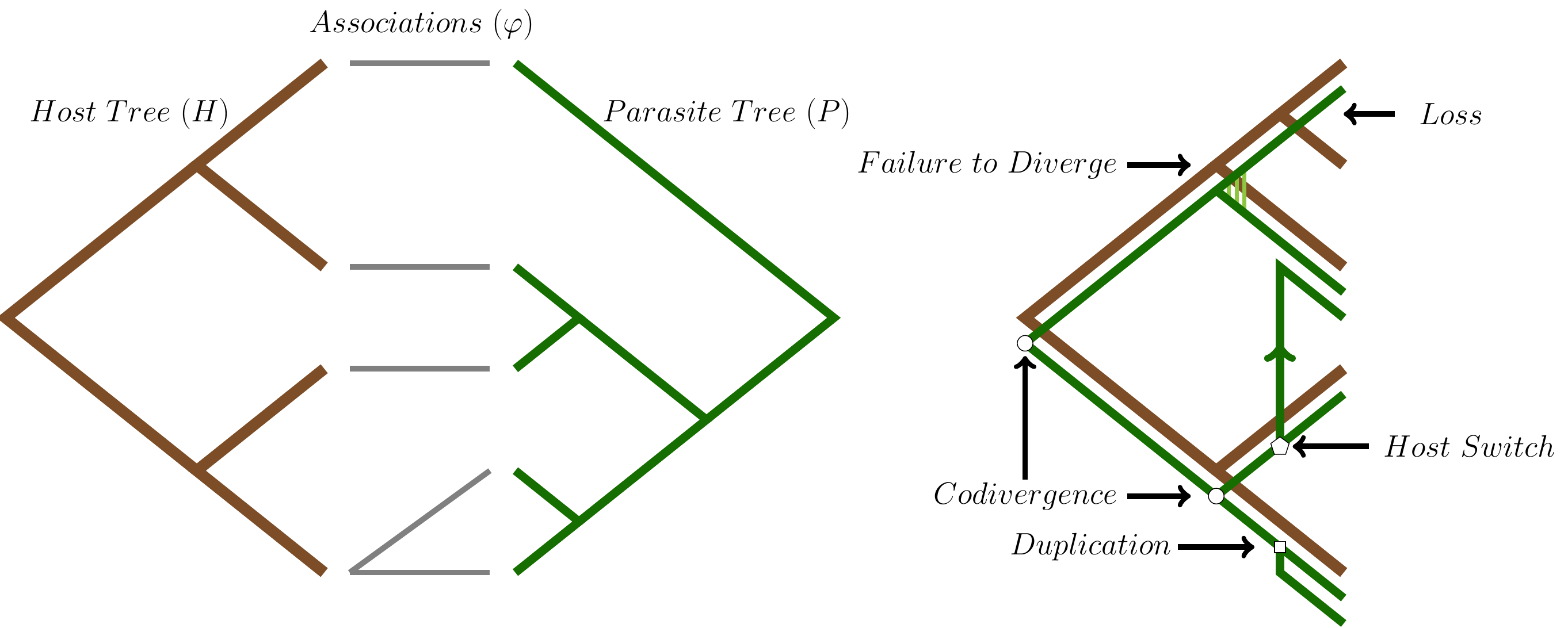}

\caption{A tanglegram (left) and one of its optimal maps (right). What is unique about this possible map, $\Phi$, is that it includes all five evolutionary events applied within current cophylogeny mapping algorithms including Jane, CoRe-PA and CoRe-ILP.}
\label{fig:TanglegramAndMap}
\end{figure}

\newpage

\begin{figure}[H]
\centering

\includegraphics[width=\textwidth]{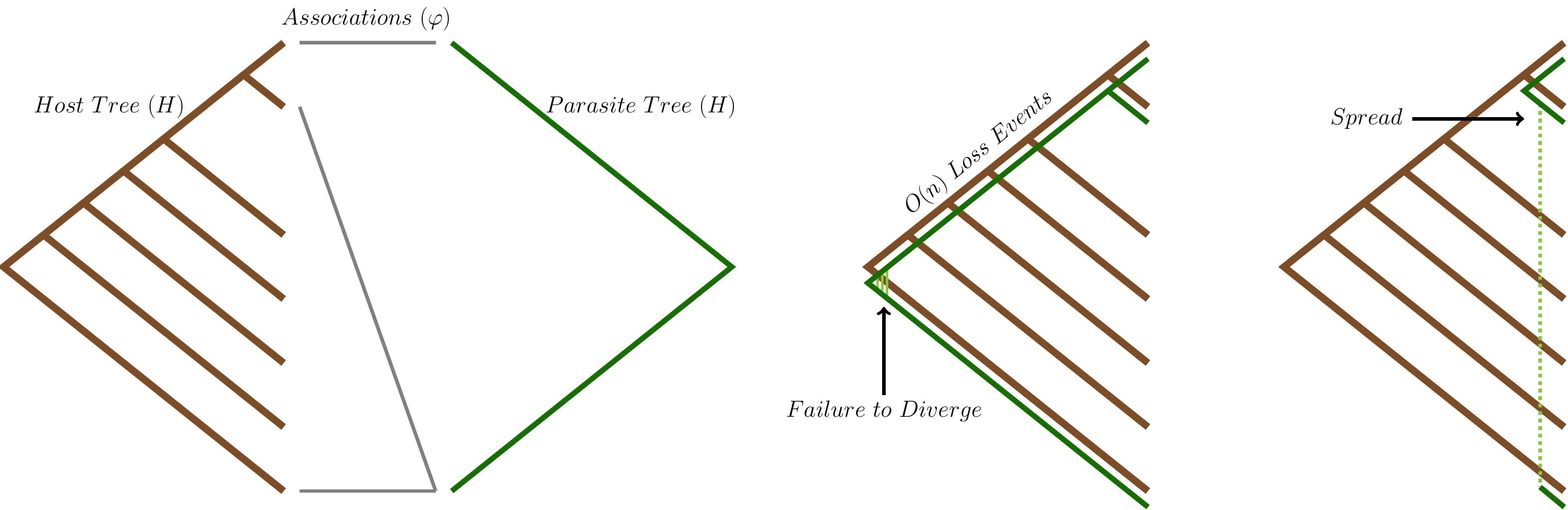}

\caption{A tanglegram (left) and two Pareto optimal solutions using either failure-to-diverge (center) or spread (right).}	
\label{fig:introducingSpread}
\end{figure}

\newpage

\begin{figure}[H]
\centering

\includegraphics[width=\textwidth]{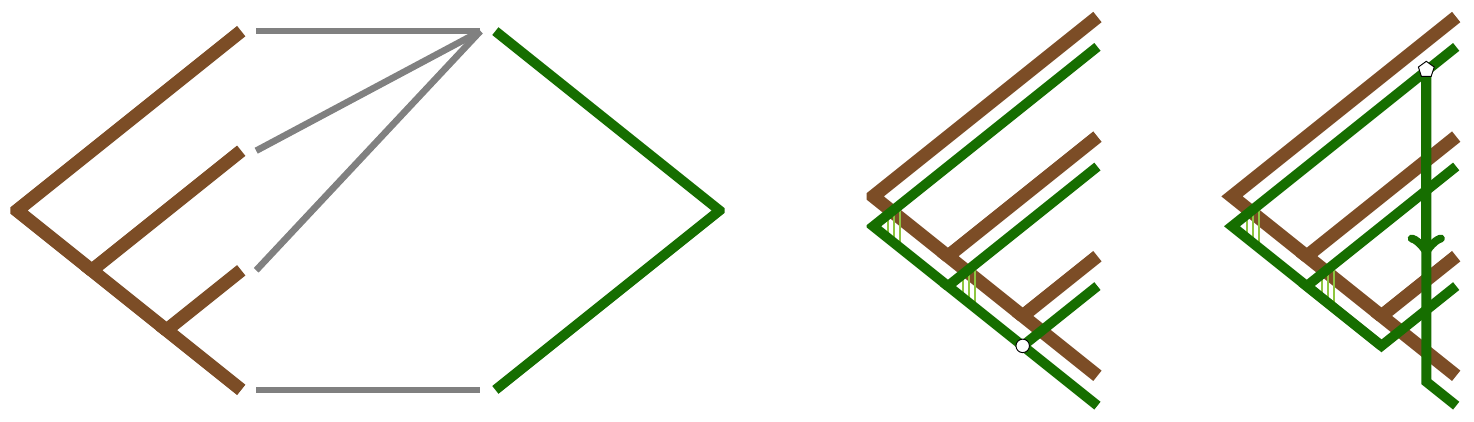}

\caption{Tanglegram which will demonstrate why current implementations of widespread parasitism reconciliation using cophylogeny mapping fail (left), and two recovered maps which include the map recovered from Jane (right), and an optimal map (center). The algorithm presented herein is the first method proposed capable of presenting an algorithmic solution capable of recovering the optimal map for this tanglegram.}

\label{fig:caseWhereJaneFails}

\end{figure}

\newpage

\begin{figure}[H]
\centering

\includegraphics[width=\textwidth]{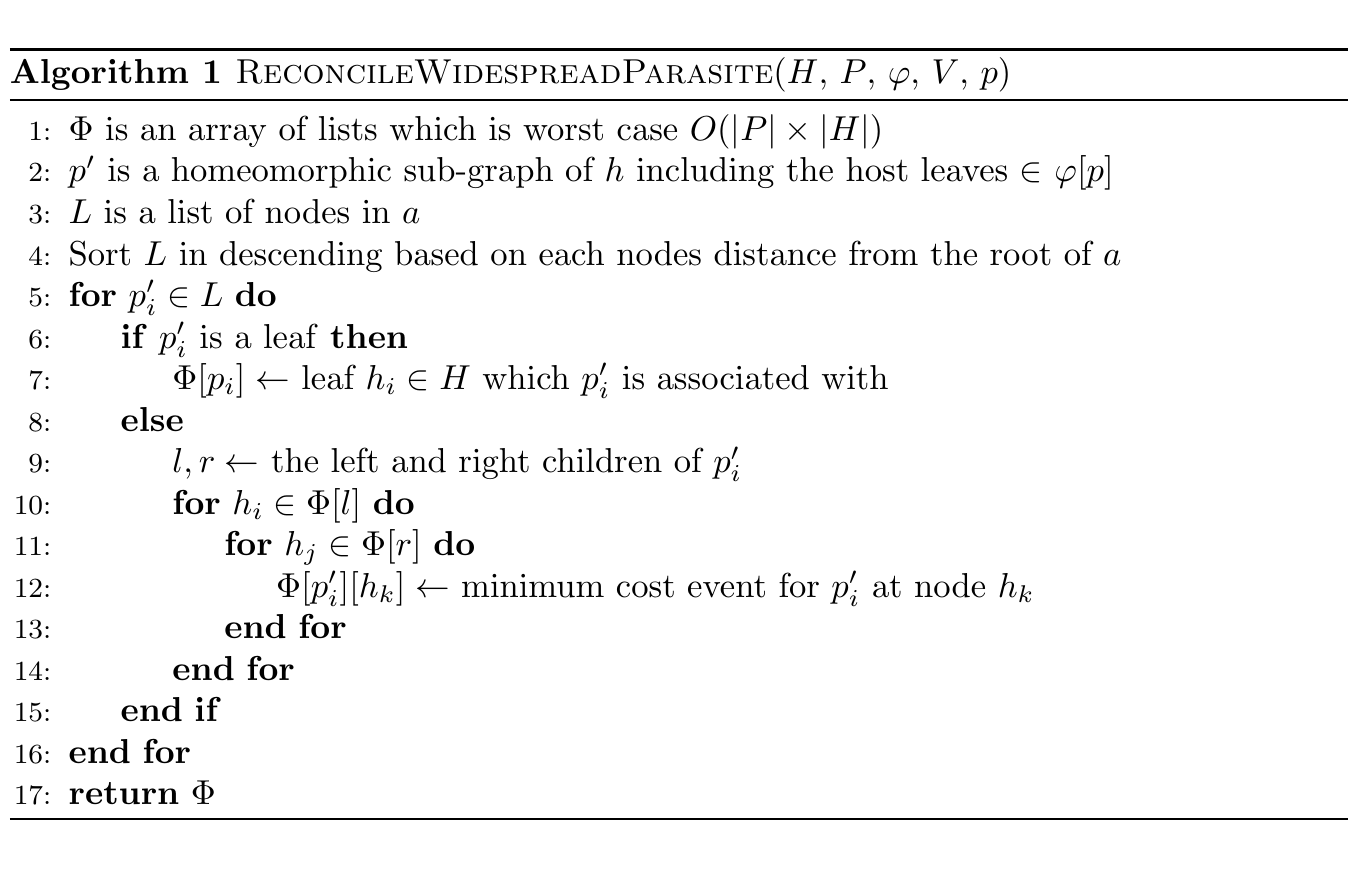}

\caption{The ReconcileWideSpreadParasite subroutine called from the WiSPA algorithm (see Figure~\ref{fig:wISPA}). This method outlines the process to infer the optimal set of widespread events from a set of association trees, $A$.}

\label{fig:reconcile}

\end{figure}

\newpage

\begin{figure}[H]
\centering

\includegraphics[width=\textwidth]{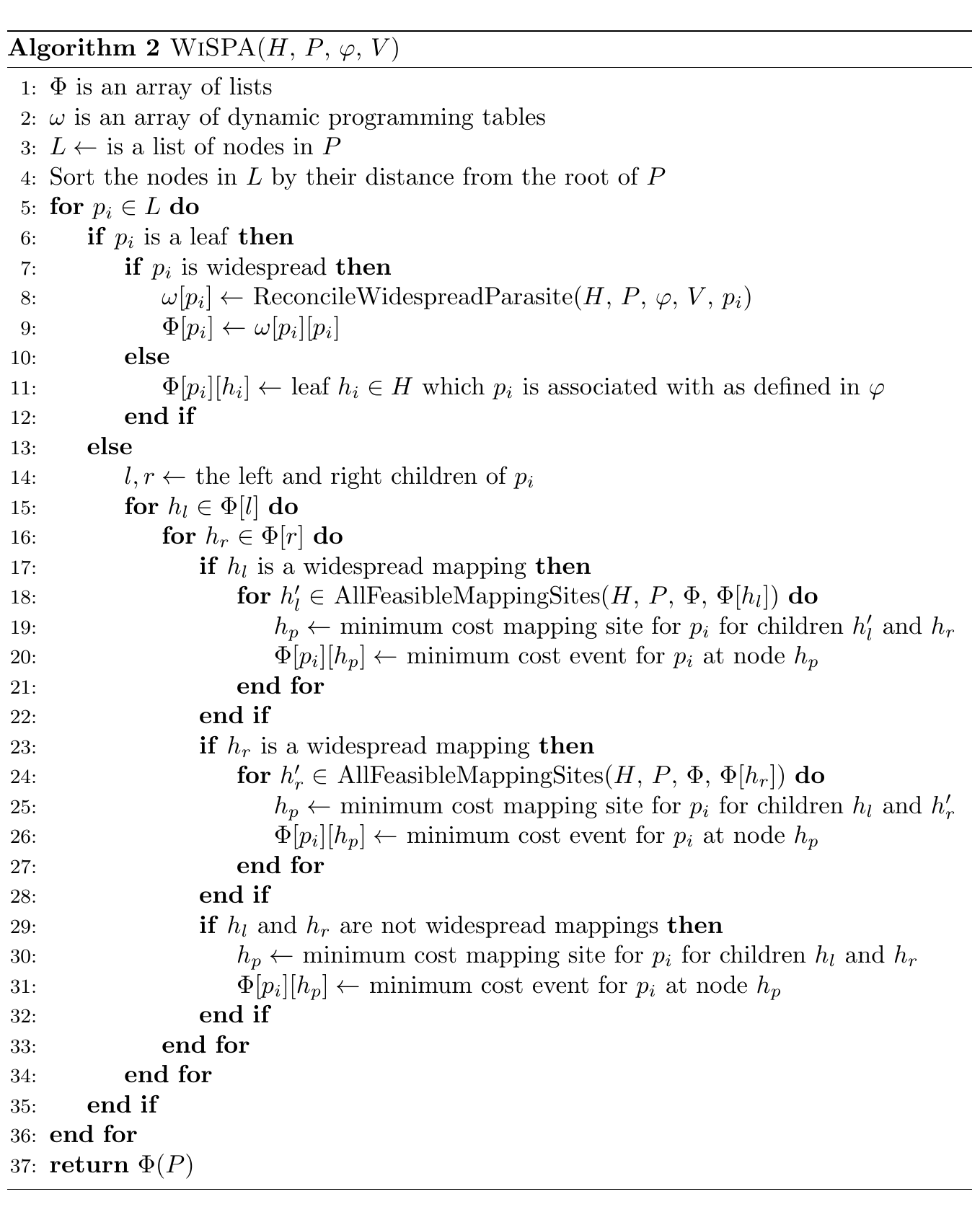}

\caption{The WiSPA algorithm which outlines the process for reconciling the optimal set of widespread and divergence events for a pair of phylogenetic trees ($H$ and $P$), based on the known associations $(\varphi)$ between the two trees.}

\label{fig:wISPA}

\end{figure}

\newpage

\begin{figure}[H]
\centering

\includegraphics[width=\textwidth]{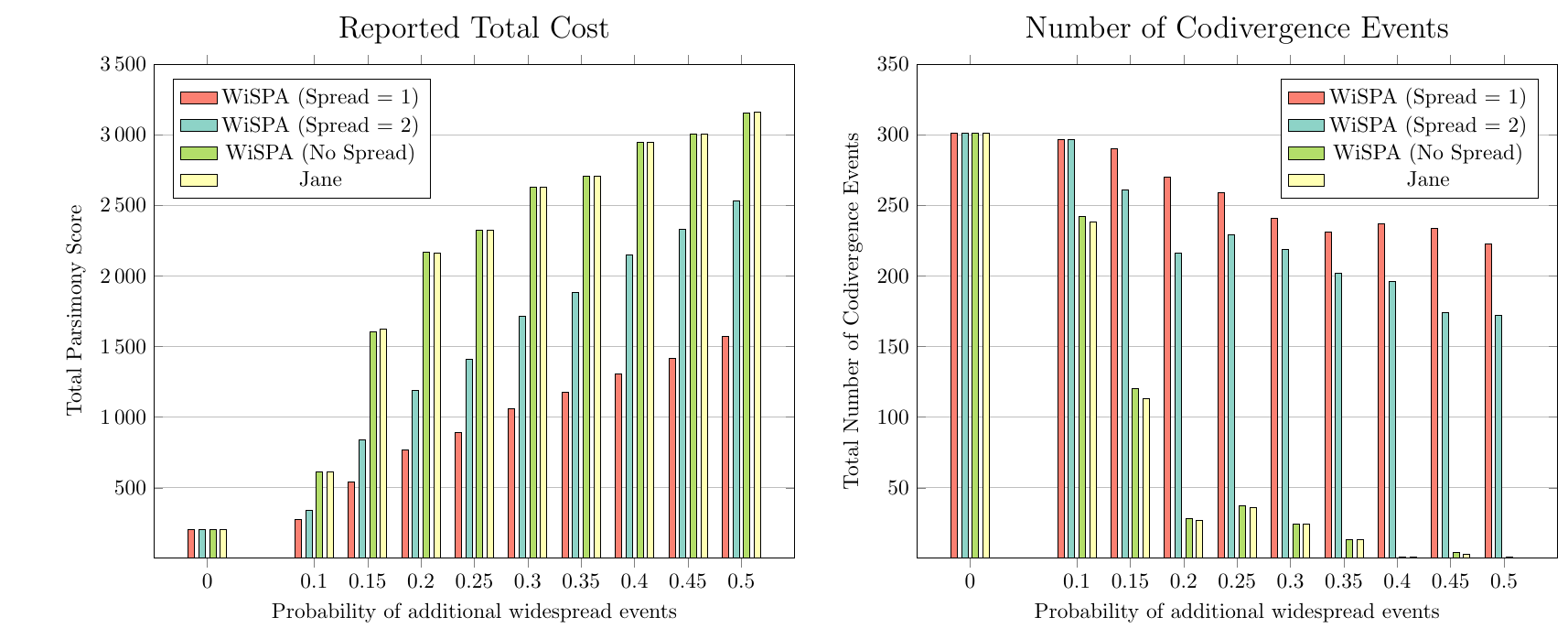}

\caption{The results for the synthetic data sets. The first plot (left) considers the rate at which the total cost over 50 synthetic coevolutionary models increases as the rate of widespread parasitism is increased, where the second plot (right) considers the rate at which the total number of codivergence events over 50 synthetic coevolutionary models decreases as the rate of widespread parastism is decreased.}

\label{fig:resultsForMetaheuristicPlots}

\end{figure}

\newpage

\begin{figure}[H]
\centering

\includegraphics[width=\textwidth]{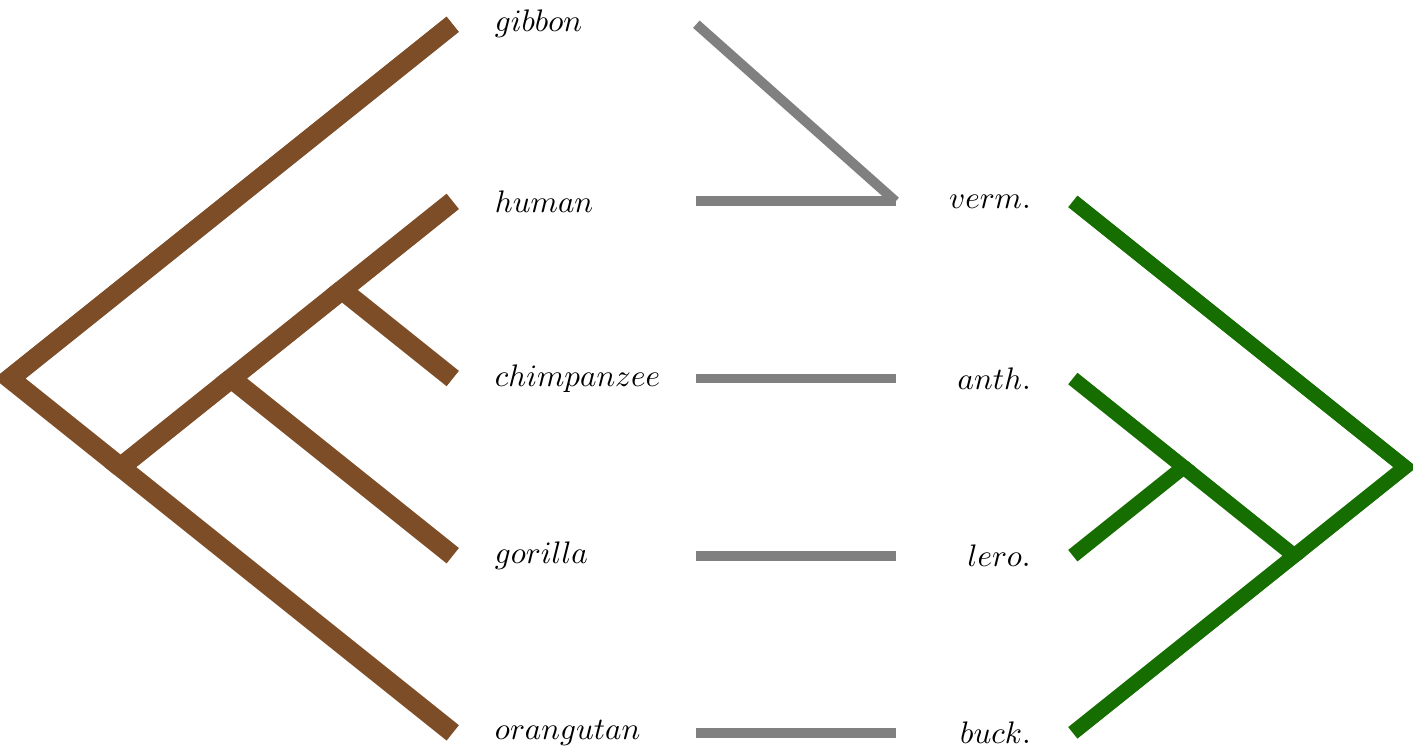}

\caption{\emph{Priamte}--\emph{Enterobius} tanglegram adapted from \citeauthor{brooks1982pinworms} \citeyearpar{brooks1982pinworms}, and \citeauthor{ronquist1997phylogenetic} \citeyearpar{ronquist1997phylogenetic} where the widespread associations are marked in red.}

\label{fig:primatePinwormsTanglegram}

\end{figure}

\newpage

\begin{figure}[H]
\centering

\includegraphics[width=\textwidth]{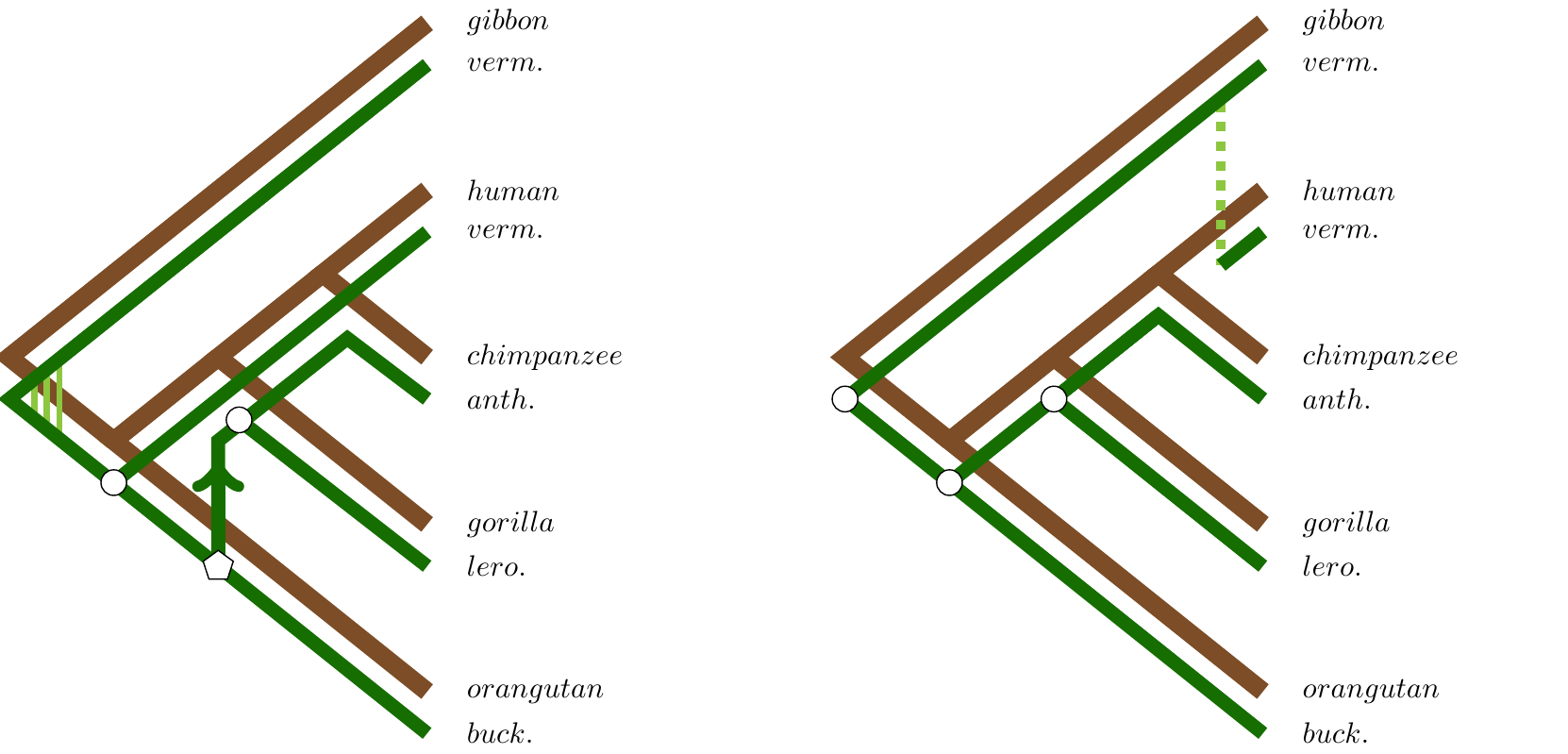}

\caption{Two optimal maps recovered for the Primate / Pinworms data set. The first map (left) is the optimal reconstruction inferred using Jane, while the second map (right) is the optimal reconstruction inferred by WiSPA.}

\label{fig:primatePinwormMaps}

\end{figure}

\newpage

\begin{figure}[H]
\centering

\includegraphics[width=\textwidth]{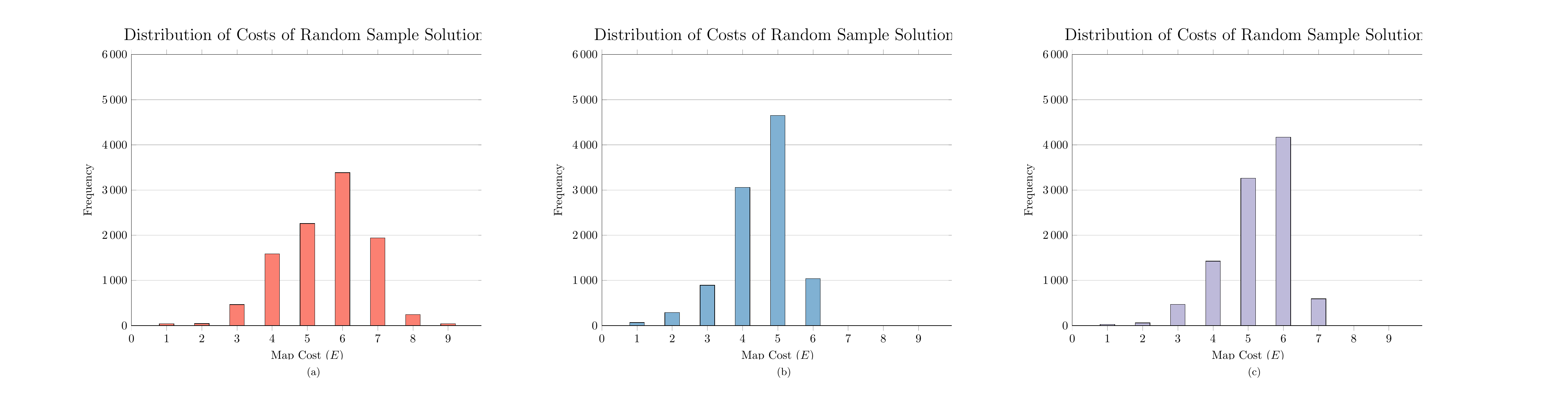}

\caption{Results from the Bernoulli trials where 10000 replicates were run. Plot (left) records the distribution of the optimal reconstruction inferred using Jane while plot (center) and (right) record the distribution of the optimal reconstruction inferred by WiSPA for the cost scheme $V=(0,1,2,1,1,1)$ and $V=(0,1,2,1,1,2)$ respectively.}

\label{fig:resultsForBernoulliTrails}

\end{figure}

\end{document}